\documentclass{article}
\usepackage[margin=1in, headheight=14.5pt]{geometry}
\usepackage{authblk}
\usepackage{amsmath,amssymb,verbatim,mathtools,amsthm,bbm,mathrsfs,xcolor}
\usepackage{natbib}
\usepackage{hyperref}
\usepackage{setspace}
\usepackage{indentfirst}
\usepackage{mathrsfs}
\usepackage{enumitem}
\newtheorem{theorem}{Theorem}

\newtheorem{lemma}{Lemma}

\newtheorem*{notation*}{Notation}
\newtheorem{corollary}{Corollary}
\theoremstyle{definition}

\newtheorem{assump}{Assumption}

\newcommand{\norm}[1]{\left\lVert#1\right\rVert}
\DeclareMathOperator{\E}{\mathbb{E}}

\DeclareMathOperator{\Prob}{\mathbb{P}}
\newcommand{\indep}{\perp \!\!\! \perp}
\title{Sequential Monte Carlo for Cut-Bayesian Posterior Computation \footnote{LA-UR: 23-31546; JM and GG contributed equally to this work; correspondence to: ggopalan@lanl.gov.}}
\author[1,2]{Joseph Mathews}
\author[2]{Giri Gopalan}
\author[2]{James Gattiker}
\author[2]{Sean Smith}
\author[2]{Devin Francom}
\affil[1]{Department of Statistical Science, Duke University, Durham, NC 27708 USA}
\affil[2]{Los Alamos National Laboratory, Los Alamos, NM 87545 USA}

\date{October 2024}
\newcommand{\X}{\mathcal{X}}

\newcommand{\y}{\mathbf{y}}

\newcommand{\h}{\theta}

\begin{document}
\onehalfspacing

\maketitle

\abstract{We propose a sequential Monte Carlo (SMC) method to efficiently and accurately compute \textit{cut-Bayesian posterior} quantities of interest, variations of standard Bayesian approaches constructed primarily to account for model misspecification. We prove finite sample concentration bounds for estimators derived from the proposed method and apply these results to a realistic setting where a computer model is misspecified. Two theoretically justified variations are presented for making the sequential Monte Carlo estimator more computationally efficient, based on linear tempering and finding suitable permutations of initial parameter draws. We then illustrate the SMC method for inference in a modular chemical reactor example that includes submodels for reaction kinetics, turbulence, mass transfer, and diffusion. The samples obtained are commensurate with a direct-sampling approach that consists of running multiple Markov chains, with computational efficiency gains using the SMC method. Overall, the SMC method presented yields a novel, rigorous approach to computing with cut-Bayesian posterior distributions.}

\section{Introduction}\label{section: intro}
Models of complex physical systems are often constructed via the coupling of multiple submodels where each represents a distinct, salient process. Submodels may entail their own experimental data, physical parameters, and a simulator of the subprocess being represented, all of which can be used for statistical modeling and inference. For example, Figure \ref{fig:reactor-parts} illustrates the structure of an ethylene-oxide reactor model consisting of coupled submodels where each submodel corresponds to a relevant physical subprocess, such as turbulence, reaction kinetics, mass transfer, and diffusion. Standard Bayesian inference, especially with the aid of hierarchical models, is a powerful way to perform uncertainty quantification within physical systems such as the ethylene-oxide reactor illustrated. Nonetheless, modern Bayesian statistics has seen the development of cut-Bayesian posteriors for modular inference \citep{liu2009, plummer2015cuts, jacob2017better, carmona2020semi, yu2021variational, frazier2022cutting}, variations of standard approaches meant primarily to dampen the pernicious effects of model misspecification. Despite many recent developments in modular inference, computation can be a major impediment for the use of cut-Bayesian posteriors and is an active area of research. Our main objective is to introduce novel sequential Monte Carlo (SMC) methods for efficiently computing with cut-Bayesian posterior distributions along with finite-sample theoretical results that underpin the concentration of resultant estimators.

\begin{figure}
    \centering
    \includegraphics[width=0.6\textwidth]{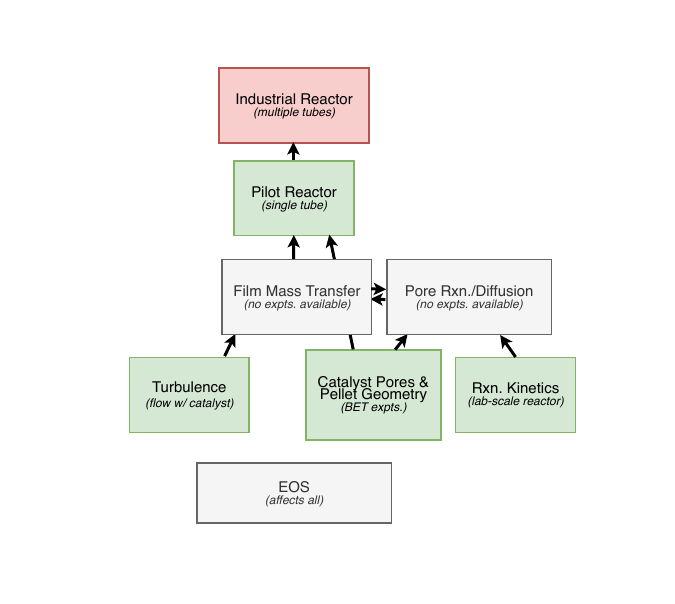}
    \caption{Components of a prototypical ethylene-oxide reactor model intended for demonstrating modular Bayesian inference, including the SMC method for cut-Bayes computation that is the subject of this paper. Green coloring indicates submodels for which experimental data exists, whereas gray corresponds to submodels that lack data. The top red model depicts the integration of multiple reactor tubes into a new design for which pure prediction would be required.}
    \label{fig:reactor-parts}
\end{figure}
While we will precisely define what is meant by a cut-Bayesian posterior distribution in the following section, the essential idea of this modeling strategy is to fix distributions over some subset of the important parameters, or to only use some subset of the available experimental data and submodels to infer a distribution over said parameters. Subsequently, the associated uncertainty of these fixed distribution parameters is integrated out through a mixture when performing standard Bayesian inference over the remaining parameters of interest. Throughout, we will refer to the parameters with a fixed distribution as \textit{cut parameters} and their corresponding fixed distribution as the \textit{cut distribution}.

While our focus is computation, we provide some of the major arguments that have been made in favor of cut-Bayesian posteriors for context. Perhaps the most commonly cited justification is that of model misspecification; some submodels could be poorly specified, and it makes sense to use data only from the well-specified submodels in inferring parameters; this approach has sometimes also been referred to as \textit{feedback cutting}. Second, it is common practice to use a plug-in point estimate for one or more parameters when there are many parameters to infer, as a way to mitigate complexity of the problem. However, the plug-in method inherently ignores uncertainty of the fixed parameter – instead, fixing a distribution over the parameter via cut Bayes provides a viable alternative to a full Bayesian approach that incorporates uncertainty for all parameters of interest. While fixing a parameter that does not provide much sensitivity could be appropriate, the same does not hold for sensitive parameters. A third, related motivation is computational complexity; for instance, a Markov chain Monte Carlo (MCMC) method may mix substantially better when working with fewer parameters than a full parameter set, for instance by avoiding non-identifiability. These and additional motivations are explicated in the work of \cite{liu2009}, \cite{plummer2015cuts}, and \cite{carmona2020semi}, amongst others.

Despite sound justification for the use of cut-Bayes methodology, sampling from or approximating the cut-Bayesian posterior distribution can be an onerous computational difficulty. As outlined in \cite{plummer2015cuts}, the ideal standard is to run many Markov chains long enough to explore the posterior distribution conditional on draws of the cut parameters, an approach we will refer to as \textit{direct sampling}. This procedure yields exact samples from a cut-Bayesian posterior distribution in contrast to sampling from cut distributions within an MCMC routine -- what \cite{plummer2015cuts} refers to as \textit{naive cut}, which is not guaranteed to converge to the target distribution. However, running many Markov chains could take an unrealistic amount of time to complete, even in a parallel environment, given how many iterations may be needed to mix for even a single chain. As an alternative, \cite{plummer2015cuts} puts forth a heuristic approach dubbed \textit{tempered cut} which takes small steps at each iteration for the cut parameter.

Besides the work of \cite{plummer2015cuts}, there is some recent literature on computing or approximating cut-Bayesian posterior distributions. One such contribution is the unbiased MCMC method of \cite{jacob2020}, which is parallel and relies on maximal couplings of Markov chain proposals to ensure distinct chains meet. Additionally, \cite{liu_goudie} present a method for approximating the cut posterior distribution that improves on the tempered cut method of \cite{plummer2015cuts} and relies on stochastic approximation Monte Carlo \citep{liang07}. \cite{yu2021variational} and \cite{carmona2022scalable} have developed variational inference for cut-Bayesian posterior computation. \cite{balde2023nonparametric} use Gaussian-process emulation to approximate the posterior distribution conditional on cut parameters along with a linearity assumption for a coupled model.

We do not believe that the SMC methods presented are direct competitors to the methods cited. For instance, it is conceivable that variational inference methods could be more computationally efficient, particularly when a mean-field approximation is feasible, due to the avoidance of Monte Carlo sampling (which can be computationally demanding). However, there are often cases where Gaussian approximations are not meaningful and our methods may be preferred in such an instance. While we are not unique in the use of theoretical support, other treatments have relied upon asymptotics (which is quite valuable, but fundamentally different from the finite-sample approach we take). In comparison, to our knowledge, the methods presented are the first developments of finite-sample theory and application of SMC for computing with a cut-Bayesian posterior distribution. Some main features of the SMC methods presented are that they:

\begin{enumerate}
    \item yield estimators that concentrate around the true cut-Bayes posterior quantities with specific finite sample bounds,
    \item do not require Gaussian or variational approximations to the posterior distribution conditional on cut parameters, and
    \item do not assume linearity of functions, for instance, in the simulator.
\end{enumerate}

The paper is structured as follows. Section \ref{section: cut-Bayes} provides formal background for cut-Bayesian posterior distributions, SMC in general, and computationally efficient extensions for computing cut-Bayesian posterior expectations (i.e., linear tempering and permuting the ordering of cut parameter draws); Theorem \ref{thm:Main Theorem Concentration} and Corollary \ref{thm:Tempered Cut} provide finite-sample concentration results for the associated estimators of these methods. Moreover, Section \ref{section: cutBayesex} includes a detailed treatment of a specific instance of cut Bayes that could arise for a misspecified computer model. For this example, the posterior distribution of calibration parameters conditional on the cut parameters is normal with a mean prescribed by a general function of the cut parameters and data; finite sample complexity results are given for this specific case. The application of SMC to a cut-Bayesian posterior in the ethylene-oxide reactor modular example is presented in Section \ref{section: Application} as an illustration on an actual scientific model. We review our main findings and suggest future avenues for work in Section \ref{section: conclusion} and conclude. 

\section{Background, SMC Method, and Concentration}\label{section: cut-Bayes}
We begin by reviewing the definition of cut-Bayesian posteriors and SMC methodology in general. Then we set forth a specific SMC method for computing with cut-Bayesian posteriors and provide concentration results for SMC-based estimators of cut-Bayesian posterior quantities. We then show how a linear tempered variant of the SMC method can be employed, as well as how different permutations of cut parameters can be used, both of which can reduce computational burden. In the subsequent section we follow up with an example of how the theoretical results can be applied to a setting where there is model misspecification, for instance with computer models in which the conditional posterior takes a Gaussian form. 

\subsection{Cut-Bayesian Posteriors}\label{subsection: cut-Bayesian posteriors}
Consider the set of parameters $(\nu,\theta)$ where $\nu \in \X^{d_{\nu}}$ and $\h \in \X^{d}$ (typically, $\X$ is $\mathbb{R}$). Let $\y = (y_{1},\ldots,y_{n})$ and suppose
\begin{align*}
    \nu \sim p_{\nu}(\cdot) \quad\quad \y \mid \nu,\h \sim p_{y}(\cdot \mid \nu, \h).
\end{align*}
Throughout, we refer to $\nu$ as the \textit{cut parameter(s)} and $\theta$ as the \textit{calibration parameter(s)}. Often $p_{\nu}$ is presented as a posterior distribution corresponding to a separate submodel, though $p_{\nu}$ could be provided by domain experts without explicit reference to a posterior distribution or experimental data. We refer to $p_{\nu}$ as the \textit{cut distribution}. The standard Bayesian posterior is given by:
\begin{align}\label{eqn:Posterior}
    \pi(\nu,\theta \mid \mathbf{y}) = \pi(\nu \mid \mathbf{y})\pi(\theta \mid \mathbf{y}, \nu) = \frac{p_{y}(\y \mid \nu,\theta)p_{\h}(\theta \mid \nu)p_{\nu}(\nu)}{m(\mathbf{y})},
\end{align}
where $p_{\theta}$ denotes a prior distribution for $\h$ and $m(\mathbf{y})$ is the marginal distribution for $\mathbf{y}$. Typically $\h$ and $\nu$ are assumed independent apriori in which case $p_{\h}(\h \mid \nu) = p_{\h}(\h)$. When $p_{y}$ is misspecified, updating $\nu$ using $\y$ can lead to poor estimates of $\h$ \citep{liu2009, plummer2015cuts}. This is a primary reason we may choose to not update $\nu$ using $\y$, leading to the \textit{cut-Bayes posterior}:
\begin{align}\label{eqn:Cut Posterior}
    \pi^{cut}(\nu,\theta \mid \mathbf{y}) = p_{\nu}(\nu) \pi(\theta \mid \y, \nu) = \frac{ p_{y}(\y \mid \nu,\theta)p_{\h}(\theta \mid \nu) p_{\nu}(\nu)}{p_{y}(\y \mid \nu)}.
\end{align}
The distributions \eqref{eqn:Posterior} and \eqref{eqn:Cut Posterior} differ due to the difference in normalizing constants, with the term $p_{y}(\y \mid \nu)$ depending on $\nu$. The dependence of the normalizing constant on $\nu$ in the cut-Bayesian posterior prohibits sampling from \eqref{eqn:Cut Posterior} using standard MCMC methods (e.g. Metropolis-Hastings).

The direct-sampling approach from \eqref{eqn:Cut Posterior} (referred to as multiple imputation in \cite{plummer2015cuts}) is accomplished by first sampling $\nu \sim p_{\nu}$ and then sampling from the conditional distribution $\pi(\h \mid \y, \nu)$. If an MCMC method is used to sample from the conditional posterior, obtaining samples from \eqref{eqn:Cut Posterior} involves running $S$ distinct Markov chains conditional on different $\nu$ values. However, such a sampling strategy can quickly become infeasible if draws from $\pi(\h \mid \y, \nu)$ are expensive, for instance due to a long time for each Markov chain to mix. We propose an SMC method to sample from the cut-Bayesian posterior efficiently. 

\subsection{SMC}\label{subsection:SMC}
SMC methods (``particle filters") are a class of algorithms designed to sequentially sample from a sequence of distributions $\mu_{0},\ldots,\mu_{S}$. For example, in filtering problems $\mu_{s}$ may be the posterior distribution of the state of a hidden Markov model conditional on the observation process. In Bayesian inference problems, $\mu_{S}$ may be the posterior distribution of a static parameter and $\mu_{s} \propto \mu^{\beta_{s}}_{S}$ for a sequence of inverse temperatures $\beta_{0} < \beta_{1} < \ldots < \beta_{S} = 1$. Asymptotic properties of SMC algorithms have been studied extensively \citep{chopin_clt,adaptive_clt}. Finite sample properties of SMC methods were given more recently in \cite{marion}. SMC algorithms are appealing in that they provide natural estimates of normalizing constants and can be adapted to parallel computing environments easily \citep{geweke}. In addition, they have been shown to perform well even in the presence of multimodality \citep{jasra,mathews}. An SMC algorithm is initialized by first obtaining samples (called \textit{particles}) from $\mu_{0}$. Having obtained particles approximately drawn from $\mu_{s-1}$, a combination of resampling and MCMC methods are used to propagate particles forward to obtain samples approximately drawn from $\mu_{s}$. Critical to the success of SMC algorithms is the ``closeness" of adjacent distributions $\mu_{s-1}$ and $\mu_{s}$, and defining a sequence of distributions that enables efficient sampling can be challenging. 

\subsection{Cut-Bayes SMC Method}\label{subsection:cut-Bayes SMC}
Here, $\mu_{s} = \pi(\h \mid \y, \nu_{s})$, where $\nu_{s}$ is a sample from the cut distribution $p_{\nu}$. Consequently, $\mu_{s}$ is itself random, but conditional on cut parameter draws, the SMC algorithm discussed below does not consider the randomness of $\nu_{s}$.  Our approach is to leverage concentration properties of $\nu_{0:S}$ so that particles approximately drawn from $\pi(\h \mid \y, \nu_{i})$ provide a close approximation of $\pi(\h \mid \y, \nu_{j})$ for $i \neq j$. In the following, we give precise conditions under which this SMC algorithm provides an efficient stochastic approximation method for estimating expectations under $\pi^{cut}$.

We now introduce the SMC cut-Bayes method. Assume that we can obtain
\begin{align}\label{eqn:Sample Nu}
    \nu_{0:S} := (\nu_{0},\ldots,\nu_{S}) \overset{iid}{\sim} p_{\nu}.
\end{align}
We note that in practice \eqref{eqn:Sample Nu} may only be ``approximately" achieved by running a Markov chain targeting $p_{\nu}$ for a sufficient length of time. In our application (Section \ref{section: Application}), $p_{\nu}$ is a known distribution (e.g. uniform) where \eqref{eqn:Sample Nu} can be achieved exactly. The goal is to efficiently sample from $\mu_{0},\ldots,\mu_{S}$, where $\mu_{s}(\h) := \pi(\theta \mid \y, \nu_{s})$. Write $\mu_{s}(\theta) = q_{s}(\theta) / Z_{s}$, where $Z_{s}$ denotes the normalizing constant of $\mu_{s}(\h)$. Denote the \textit{importance weights} $w_{s}(\h) := q_{s}(\h) / q_{s-1}(\h)$ and define their corresponding empirical average as
\begin{align*}
    \bar{w}_{s}(\theta^{1}_{s-1},\ldots, \h^{N}_{s-1}) = \bar{w}_{s} := \frac{1}{N}\sum^{N}_{i=1} w_{s}(\theta^{i}_{s-1}).
\end{align*}
Conditional on $\nu_{0:S}$, the algorithm we consider is a standard SMC algorithm: 
\begin{enumerate}
    \item Sample $\nu_{0:S} \overset{iid}{\sim} p_{\nu}$
    \item Sample $\theta^{1:N}_{0} := (\h^{1}_{0},\ldots,\h_{0}^{N}) \overset{iid}{\sim} \mu_{0}$
    \item For $s=1,\ldots,S$:
    \begin{enumerate}
        \item Sample $\tilde{\h}^{1:N}_{s} := (\tilde{\h}^{1}_{s},\ldots,\tilde{\h}^{N}_{s})$, where
        \begin{align*}
            \tilde{\h}^{i}_{s} = \h^{j}_{s-1} \text{ with probability } \frac{w_{s}(\theta^{j}_{s-1})}{N\bar{w}_{s}}
        \end{align*}
        \item Sample $\theta^{1:N}_{s} := (\h^{1}_{s},\ldots,\h^{N}_{s})$, where
        \begin{align*}
            \h^{i}_{s} \sim K^{t}_{s}(\tilde{\h}^{i}_{s}, \cdot),
        \end{align*}
    \end{enumerate}
\end{enumerate}
where $K_{1},K_{2},\ldots,K_{S}$ are a set of Markov kernels assumed to be ergodic and $K_{s}$ is $\mu_{s}$-invariant. The estimate of $\E_{\pi^{cut}}[g]$ is given by
\begin{align}\label{eqn:Estimator}
    \hat{g}_{\text{SMC}} := \frac{1}{S+1}\sum^{S}_{s=0} \left(\frac{1}{N}\sum^{N}_{i=1} g(\nu_{s},\theta^{i}_{s})\right),
\end{align}
for a measurable function $g(\h,\nu)$. The following theorems give conditions under which $\hat{g}_{\text{SMC}}$ concentrates around $\E_{\pi^{cut}}[g]$ for $g$ such that $|g| \leq 1$. Before stating these, we first fix some notation. Going forward, we let $\Prob$ and $\E$ generically denote the probability measure with respect to the particle system and $\nu_{0:S}$ (see Appendix~\ref{app:A} for details). Let $\mathcal{P}_{M}(\mu_{s})$ with $M > 0$ be the set of all probability measures over $\X^{d}$ such that
\begin{align*}
    \sup_{B \subset \X^d} (\eta(B) / \mu_{s}(B)) \leq M, \text{ for } \eta \in \mathcal{P}_{M}(\mu_{s}).
\end{align*}
A distribution $\eta \in \mathcal{P}_{M}(\mu_{s})$ is said to be \textit{M-warm} with respect to $\mu_{s}$. Define the mixing time
\begin{align*}
    \tau_{s}(\epsilon, M) := \min\; \left\{t : \sup_{\eta \in \mathcal{P}_M(\mu_{s})} \|\eta K_{s}^t - \mu_{s}\|_{\text{TV}}\leq \epsilon \right\},
\end{align*}
where $\norm{\cdot}_{\text{TV}}$ denotes the total variation norm. Warm mixing times are commonly studied in the mixing time literature \citep{vempala_geometric_random_walks, yuansi}. Roughly speaking, $\eta \in \mathcal{P}_{M}(\mu_{s})$ implies $\eta$ is ``close" to $\mu_{s}$ in that their density functions are within a constant factor of each other. \cite{marion} showed that $\mathcal{L}(\tilde{\h}^{i}_{s})$ is a $2$-warm start with respect to $\mu_{s}$ for appropriate choices of $N$ and $t$ ($\mathcal{L}$ refers to the probability law of a random variable). However, obtaining this warm start requires $N$ to grow linearly in $\max_{s \in \{1,\ldots,S\}}\chi^{2}(\mu_{s} \mid \mid \mu_{s-1})$, where
\begin{align*}
    \chi^{2}(\mu_{s} \mid \mid \mu_{s-1}) := \int_{\mathcal{X}}\left(\frac{\mu_{s}(dx)}{\mu_{s-1}(dx)} - 1 \right)^{2} \mu_{s-1}(dx), 
\end{align*}
denotes the $\chi^{2}$-divergence between $\mu_{s}$ and $\mu_{s-1}$.
Here, we have less control on this quantity since the sequence $\mu_{0},\ldots,\mu_{S}$ is indexed by $\nu_{0:S}$ and is therefore defined randomly at the beginning of the algorithm. For this reason, we make the following assumption: 
\begin{assump}\label{assump:One}
There exist $\mathcal{E}_{\alpha} > 1$ and $\delta_{\alpha} \in (0,1)$ such that 
\begin{align*}
    \Prob\left(\max_{s \in \{1,\ldots,S\}}\chi^{2}(\mu_{s} \mid \mid \mu_{s-1}) \geq \mathcal{E}_{\alpha} - 1\right) \leq \delta_{\alpha}.
\end{align*}
\end{assump}
The following theorem states conditions under which $\hat{g}_{SMC}$ approximates expectations under $\pi^{cut}$.
\begin{theorem}\label{thm:Main Theorem Concentration}
Suppose $\nu_{0:S} \overset{iid}{\sim} p_{\nu}(\cdot)$ and $\theta^{1:N}_{0} \overset{iid}{\sim} \pi(\cdot \mid \y, \nu_{0})$.  Let $\delta \in (0,\frac{1}{4})$, $\epsilon \in (0,1)$, and set
\begin{enumerate}
    \item $N \geq \log\left(\frac{6(S+1)}{\delta} \right) \max\left\{18 \mathcal{E}_{\alpha}, \frac{2}{\epsilon^{2}} \right\}$
    \item $t \geq \max_{s=1,\ldots,S} \tau_{s}(\frac{\delta}{6N(S+1)}, 2)$
    \item $S \geq \frac{2}{\epsilon^{2}}\log\left(\frac{6}{\delta} \right)$
\end{enumerate}

Then under Assumption~\ref{assump:One}, 
\begin{align*}
    \Prob\left(\left|\hat{g}_{\text{SMC}} - \E_{\pi^{\text{cut}}}[g] \right| \leq \epsilon \right) \geq 1 - \delta - \delta_{\alpha}, \text{ for } |g| \leq 1.
\end{align*}
\end{theorem}
The proof of Theorem~\ref{thm:Main Theorem Concentration} is given in Appendix~\ref{app:A}. Our approach is to apply known finite sample results for SMC samplers given in \cite{marion}. The results of \cite{marion} apply immediately upon conditioning on $\nu_{0:S}$ since in this case our algorithm is identical to the one considered in their results. The key difference is that here $\mu_{0},\ldots,\mu_{S}$ are randomly generated at the beginning of the algorithm. The resulting cost to ensure \eqref{eqn:Estimator} concentrates around $\E_{\pi^{cut}}[g]$ is that we now stipulate $S = \mathcal{O}(\epsilon^{-2})$. This is because the particles at a given step of the algorithm provide an estimate of the conditional expectation $\E[g(\h, \nu_{s}) \mid \nu_{s}]$ and $S = \mathcal{O}(\epsilon^{-2})$ ensures that these conditional expectations concentrate around their mean $\E_{\pi^{cut}}[g]$, based on Hoeffding's inequality.

For the concentration bound of Theorem~\ref{thm:Main Theorem Concentration} to be useful, it is important  to ensure that $\mathcal{E}_{\alpha}$ is not very large for choice of $\delta_{\alpha}$, since the number of particles, $N$, is linear in  $\mathcal{E}_{\alpha}$. Depending on the form of the conditional posterior distribution $\mu_s$ and the choice of $S$, it is possible that a small $\mathcal{E}_{\alpha}$ is achievable for a given $\delta_{\alpha}$; i.e., the conditional posteriors tend not to be far apart for the $S+1$ cut parameters sampled. However, a variation of the above, which we call \textit{tempered cut-Bayes SMC}, can be used to reduce $\mathcal{E}_{\alpha}$ for a given $\delta_{\alpha}$, with essentially the same concentration bounds and slight modification. We next introduce this variant to aid in cases where the $S+1$ conditional posterior distributions tend to be far apart (in the $\chi^2$ distance sense). 

\subsection{Tempered Cut-Bayes SMC Method}\label{subsection: temperedSMC}
The proposed tempered cut-Bayes SMC method is related to linear tempering presented in Section 5 of \cite{plummer2015cuts}, but the heuristic of \cite{plummer2015cuts} differs by using MCMC instead of SMC. Without essential loss of generality, assume that the cut parameters are elements of $\mathbb{R}^{d_{\nu}}$ (see \cite{plummer2015cuts} for the case of discrete cut parameters). Consider the same set up as in Section \ref{subsection:cut-Bayes SMC} except augment the cut parameter sequence by connecting two consecutive and independently drawn cut parameter draws, $\nu_{s}$ and $\nu_{s+1}$, with a straight line and adding $P$ evenly spaced cut parameters in between. Consider performing SMC on this newly constructed sequence of $(P+1)S+1$ conditional posterior distributions, which presumably will tend to be closer together in $\chi^2$ distance than the original $S+1$ conditional posteriors, achieving a smaller $\mathcal{E}_{\alpha}$ for a given $\delta_{\alpha}$. If we replace $(P+1)S$ for $S$ in the choice of $N$ and $t$ in Theorem 1, then, by the proof of Theorem 1 in the Appendix, we will have a coupling of particles with high probability for all of the conditional posteriors, and thus high probability of coupling for the conditional posteriors indexed only by the $S+1$ independently drawn cut parameters. Hence, we may compute the estimator as in Equation 4, retaining only the particles at $S+1$ cut parameter draws, and the proof follows as for Theorem 1.

To summarize, without essential loss of generality assume that cut parameters are elements of $\mathbb{R}^{d_{\nu}}$, and construct a new sequence $\nu^*_{0:(P+1)S}$ that augments $\nu_{0:S}$ with $P$ equally spaced points along the line connecting consecutive points $\nu_{s}$ and $\nu_{s+1}$ for $s \in \{0,...,S-1\}$. Index the conditional posteriors corresponding to $\nu^*_{0:(P+1)S}$ as $\mu^*_s(.)$. Then, our main assumption is essentially the same as Theorem 1:
\begin{assump}\label{assump:Two}
There exist $\mathcal{E}^*_{\alpha} > 1$ and $\delta_{\alpha} \in (0,1)$ such that 
\begin{align*}
    \Prob\left(\max_{s \in \{1,\ldots,S(P+1)\}}\chi^{2}(\mu^*_{s} \mid \mid \mu^*_{s-1}) \geq \mathcal{E}^*_{\alpha} - 1\right) \leq \delta_{\alpha}.
\end{align*}
\end{assump}

\begin{corollary}\label{thm:Tempered Cut}
Suppose $\nu_{0:S} \overset{iid}{\sim} p_{\nu}(\cdot)$ and $\theta^{1:N}_{0} \overset{iid}{\sim} \pi(\cdot \mid \y, \nu_{0})$. Construct a new sequence $\nu^*_{0:(P+1)S}$ that augments $\nu_{0:S}$ with $P$ equally spaced points along the line connecting consecutive points $\nu_{s}$ and $\nu_{s+1}$ for $s \in \{0,...,S-1\}$. Let $\delta \in (0,\frac{1}{4})$, $\epsilon \in (0,1)$, and set
\begin{enumerate}
    \item $N \geq \log\left(\frac{6[(P+1)S+1]}{\delta} \right) \max\left\{18 \mathcal{E}^*_{\alpha}, \frac{2}{\epsilon^{2}} \right\}$
    \item $t \geq \max_{s=1,\ldots,S(P+1)} \tau_{s}(\frac{\delta}{6N[(P+1)S+1]}, 2)$
    \item $S \geq \frac{2}{\epsilon^{2}}\log\left(\frac{6}{\delta} \right)$
\end{enumerate}

Then under Assumption~\ref{assump:Two},
\begin{align*}
    \Prob\left(\left|\hat{g}_{\text{SMC}} - \E_{\pi^{\text{cut}}}[g] \right| \leq \epsilon \right) \geq 1 - \delta - \delta_{\alpha}, \text{ for } |g| \leq 1.
\end{align*}
\end{corollary}

Note that $\hat{g}$ is defined the same as in Theorem 1, i.e., it uses only the particles at the $S+1$ independently drawn cut parameters, not including the $P$ points between consecutive draws. The possible advantage in Corollary 1 is that $\mathcal{E}^*_{\alpha}$ can be made smaller than $\mathcal{E}_{\alpha}$ for a given $\delta_{\alpha}$, leading to a smaller $N$, with potentially larger $t$; the exact tradeoff will depend on the Markov kernel used, and Section \ref{section: cutBayesex} shows an example where the increase in $t$ is not substantial for appropriate choice of Markov kernel. In addition, the computational cost of additional resampling and mutation steps must be weighed against reduction in the number of particles, $N$. We next illustrate how permuting the initial cut-parameter draws can help reduce computational complexity via the proven finite-concentration bounds. %We next illustrate the applicability of the cut-Bayes SMC method and its linear tempered variant in the following section, in the context of a realistic computer model scenario where the computer model is misspecified. 

\subsection{Permuting SMC Sequences}\label{subsection: permutedSMC}
The proof of concentration bounds in Theorem 1 given in Appendix \ref{app:A} uses the independence assumption of $\nu_{0:S}$ in order to apply Hoeffding's inequality to show that the sample mean of conditional posterior means, $\frac{1}{S+1}\sum_{s=0}^S \mu_g(\nu_s)$, concentrates around the cut posterior mean with high probability. However, the specific ordering and independence of $\nu_{0:S}$ are not needed for the remainder of the arguments in Appendix \ref{app:A}, and thus one can permute the initial IID samples drawn, $\nu_{0:S}$, for a new sequence of $S+1$ conditional posteriors. Specifically, the bounds of Appendix \ref{app:A} hold for any permutation $\sigma$ from the symmetric group on $S$ points that yields the sequence $\nu_{0}, \nu_{\sigma(1)}..., \nu_{\sigma(S)}$. With respect to the concentration bounds of Theorem 1, an advantage of permuting the cut parameter draws is evident in a smaller value of $\mathcal{E}_{\alpha}$ for specific value of $\delta_{\alpha}$, which thus reduces the required values of $N$ and $t$ in the concentration bounds. In other words, permuting the cut parameters such that successive conditional posteriors are closer together requires a smaller number of particles and Markov chain mutate steps. Hence, permuting the sequence of cut parameters provides another means of reducing computational complexity of the SMC method in addition to linear tempering. There are a number of ways to permute the cut parameters, and we consider framing this problem as a variant of the travelling salesman problem. 

\paragraph{Travelling Salesman Problem} Finding an efficient sequence of conditional posteriors to tour through can be considered a variant of the travelling salesman problem. Consider a complete graph where the nodes are $\nu_{0:S}$ and edge weights are distances between the cut parameters. If analytically tractable one should use chi-squared distance between the conditional posteriors defined by cut parameters or else some proxy such as Euclidean distance between cut parameters (e.g., see example in Section \ref{section: cutBayesex} with multivariate normal distributions). Solving the travelling salesman problem -- finding a tour through all nodes that visits each node exactly once and minimizes the total distance travelled -- should avoid excessively large jumps in cut parameter space. While the travelling salesman problem is NP-hard, approximate heuristics are available as implementations in the \texttt{R} package \texttt{TSP} \citep{TSPref}. Since we want to start at a designated point and not return, we include a dummy node that has infinite distance to all other nodes besides the start and 0 distance to the start, and remove the dummy node to get an optimal path through cut parameters beginning at the designated start node. 

We illustrate the application to cut parameter draws in Figure \ref{fig:TSP_ex}. An independent and identically drawn sample of 25 hypothetical cut parameters is generated from a bivariate normal cut distribution. On the right is the random path generated by connecting successive randomly drawn points. On the left is the estimated shortest path found with the \texttt{TSP} \texttt{R} package \citep{TSPref} that visits each point exactly once (i.e., a Hamiltonian path). The Hamiltonian path has a maximum distance between successive nodes that is 4.73, whereas the random path has a maximum distance between successive nodes that is 12.5; thus, estimated shortest Hamiltonian path is bound to perform better for SMC than the random one. (We focus here on the maximum distance between successive nodes because the concentration bounds we have introduced depend on the maximum distance between successive conditional posterior distributions.)

To better understand the distribution of maximum consecutive distance over the 25 cut parameters, we resampled 25 cut parameters 1000 times. In the case of no permutation, the proportion of resamplings such that the maximum consecutive distance exceeded 10 was .857, but in the case of permutation only .003 of the resamplings exceeded 10. We repeated this procedure with a 10-dimensional cut-parameter space (instead of 2) and similar advantages remain -- .482 of resamplings without permutation had a maximum consecutive distance exceeding 10, whereas .025 of the resamplings with permutation had a maximum consecutive distance exceeding 10. These results suggest that permutation of cut parameters via the travelling salesman approach can help achieve tighter concentration bounds. 

\begin{figure}
    \centering
    \includegraphics[width=0.8\textwidth]{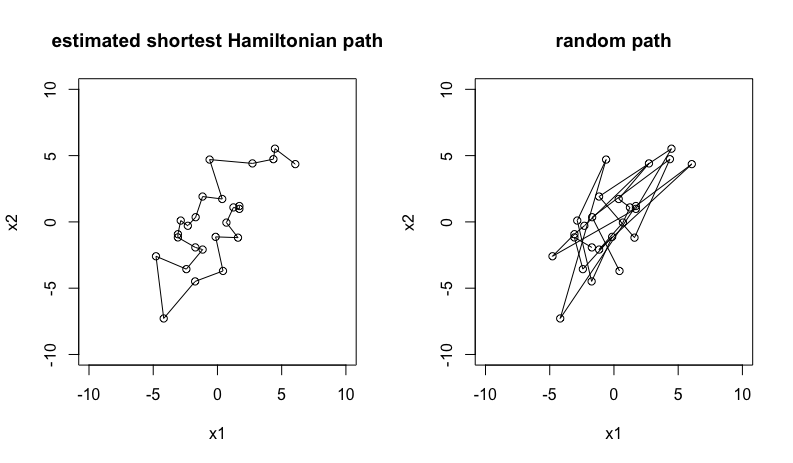}
    \caption{Left plot illustrates the estimated shortest Hamiltonian path, with a maximum successive node distance of 4.73, whereas the maximum successive node distance in the random path (right) is 12.5. There are 25 points drawn from a bivariate normal distribution with positive correlation.}
    \label{fig:TSP_ex}
\end{figure}

\section{SMC Cut Bayes for Computer Model Misspecification}\label{section: cutBayesex}
A common setting where model misspecification arises is in computer models, for instance in the Bayesian calibration literature \citep{kennedy2001bayesian, brynjarsdottir2014learning, higdon2008computer}. We consider an example where a cut-Bayes approach can be employed. Consider the set of parameters $(\nu,\theta)$ where $\nu \in \mathbb{R}^{d_{\nu}}$ and $\h \in \mathbb{R}^{d}$. Let $\y \in \mathbb{R}^{d}$. The notation $\mathcal{N}_{d}(\cdot, \cdot)$ indicates a multivariate normal density where the first argument is the mean and second is the covariance matrix. 
The model is specified as follows:
\begin{align*}
\y|\theta \sim \mathcal{N}_{d}(\theta, \sigma^2I),
\end{align*}
where $\y$ is a single realization, though without essential loss of generality there could be multiple, independent $\y$'s in the sample. A prior for $\theta$ conditional on $\nu$ is given by:
\begin{align*}
\theta|\nu  \sim \mathcal{N}_{d}(f(\nu), \sigma_P^2I),
\end{align*}
where $f$ is a general function $f: \mathbb{R}^{d_{\nu}} \rightarrow \mathbb{R}^{d} $, which could be thought of as the output of a computer simulator or model. Note that marginally:
\begin{align*}
\y|\nu  \sim \mathcal{N}_{d}(f(\nu), (\sigma^2+\sigma_P^2)I),
\end{align*}
which is similar to models often encountered in inverse problems (e.g. see Section 2 of \cite{stuart_teckentrup_2018}). A
problem with using the full Bayesian treatment for inference of $\theta$ and $\nu$ is that misspecification of $f(\cdot)$ could lead to poor inference of both parameters. With many independent samples of $\y$, it is possible to get more reliable estimates of $\theta$ (i.e., due to posterior asymptotics), but misspecification of $f(.)$ could still lead to biased and overly confident estimates of $\nu$ \citep{brynjarsdottir2014learning}. Hence, to mitigate the effect of misspecification of $f(\cdot)$, one can cut on $\nu$ and fix a cut distribution, $\nu \sim p_{\nu}$, based on auxiliary models and experiments that appropriately characterize the uncertainty for $\nu$. It should also be noted that a model discrepancy term could still be used in combination with a cut distribution for $\nu$, in order to account for computer model misspecification.  

Given the specification of the likelihood and prior conditional on $\nu$ above, the conditional posterior is given by 
\begin{align*}
\pi(\h \mid \y, \nu_{s}) = \mathcal{N}_{d}(\h; w\y+(1-w)f(\nu_{s}), cI_{d}),
\end{align*}
where $w = \frac{\sigma^{-2}}{(\sigma^{-2}+\sigma_P^{-2})}$ and $c = (\sigma^{-2}+\sigma_P^{-2})^{-1}$, by conjugacy. We show in the following arguments that the theorem and corollary of the previous section are directly applicable to analyzing a cut Bayes posterior in this model.

We consider the SMC cut-Bayes method when $f(\cdot)$ is a $\Delta$-Lipschitz function with respect to the Euclidean norm, i.e.,
\begin{align*}
    \norm{f(\nu) - f(\nu^{\prime})}_{2} \leq \Delta \norm{\nu - \nu^{\prime}}_{2}.
\end{align*}
Note that $\pi(\h \mid \y, \nu_{s})$ is strongly log-concave and log-smooth with condition number $\kappa = 1$, which can be checked by choosing $L$ and $m$ to both be $c$ in the definition of log-concave and log-smooth given in \cite{WuSchmidlerChen}. Hence, using the Metropolis-adjusted Langevin algorithm (MALA) for the MCMC kernel yields by \cite{WuSchmidlerChen}
\begin{align*}
    \max_{s=1,\ldots,S} \tau_{s}\left(\frac{\delta}{6N(S+1)}, 2\right) = \mathcal{O}^{*}(d^{\frac{1}{2}} ),
\end{align*}
where the notation $\mathcal{O}^{*}$ indicates the omission of terms logarithmic in $d$ and $\epsilon$.
By Theorem~\ref{thm:Main Theorem Concentration}, the proposed cut-Bayes SMC algorithm provides a randomized approximation algorithm for estimating $\E_{\pi^{cut}}[g]$ in time
\begin{align}\label{eqn: example cut-Bayes SMC}
    N S t = \mathcal{O}^{*}\left( \frac{d^{\frac{1}{2}}}{\epsilon^{2}} \max\{\mathcal{E}_{\alpha},\frac{1}{\epsilon^2}\}\right).
\end{align}

We calculate the $\chi^2$-divergence between consecutive conditional posteriors, $\mu_{s-1}$ and $\mu_{s}$, in the following. Note that the $\chi^2$-divergence between $\mu_{s-1}$ and $\mu_{s}$ is unchanged by subtracting off the common $w\y$ term in the mean. Further, for convenience, define $u(\cdot) := (1-w)f(\cdot)$. Then the $\chi^2$-divergence calculation proceeds as follows. 

\begin{align*}
   \chi^{2}(\mu_{s} \mid \mid \mu_{s-1}) + 1  &= \frac{1}{c^{d/2}(2\pi)^{d/2}} \int_{\mathbb{R}^{d}} \exp\left\{-\frac{1}{2c}(2\norm{\h - u(\nu_{s})}^{2} - \norm{\h - u(\nu_{s-1})}^{2}) \right\} d\theta\\
   &= \frac{1}{c^{d/2}(2\pi)^{d/2}} \int_{\mathbb{R}^{d}} \exp\left\{-\frac{1}{2c} \norm{\h - (2u(\nu_{s}) - u(\nu_{s-1}))}^{2} + \frac{1}{c}\norm{u(\nu_{s}) - u(\nu_{s-1})}^{2} \right\} d\theta \\
   &= \exp\left\{c^{-1}\norm{u(\nu_{s}) - u(\nu_{s-1})}^{2} \right\} \\
   &\leq \exp\left\{c^{-1}(1-w)^2\Delta^{2} \norm{\nu_{s} - \nu_{s-1}}^{2} \right\},
\end{align*}
where line 2 proceeds from line 1 by completing the square, and line 3 proceeds from line 2 by recognizing the integral of a Gaussian density. Now, note that using $P$ evenly spaced points between each pair of cut parameters reduces the distance between consecutive points to $\frac{\norm{\nu_{s} - \nu_{s-1}}}{P+1}$. Thus, the linear tempering SMC method reduces the term $\mathcal{E}_{\alpha}$ by an exponential factor, even for small $P$, yielding potential savings in $N$. Further, if the MALA algorithm is used, then, up to logarithmic factors, there is no increase in $t$ required based on the result cited previously from \cite{WuSchmidlerChen}.

It is still worthwhile to examine the behavior of the SMC cut-Bayes method without linear tempering. For example, suppose $\nu_{s}$ is sub-Gaussian with parameter $1 / \sqrt{\sigma}$, i.e.,
\begin{align*}
    \E[e^{t^{\top}(\nu_{s} - \E[\nu_{s}])}] \leq e^{\norm{t}^{2} / 2\sigma}.
\end{align*}
If $p_{\nu}$ is a posterior distribution, for example, then under suitable regularity conditions it will asymptotically resemble a normal and thus sub-Gaussian distribution. Then since $ \norm{\nu_{s} - \nu_{s-1}} \leq \norm{\nu_{s} - \E[\nu_{s}]} + \norm{\nu_{s-1} - \E[\nu_{s-1}]}$ we have by \cite{hsu}:
\begin{align}\label{eqn: Log Concave Conc. Bound}
    \Prob\left(\exp\left\{(1-w)^2{\Delta^{2}\norm{\nu_{s} - \nu_{s-1}}^{2}}\right\} > \exp\left\{\frac{4d(1-w)^2\Delta^{2}}{\sigma}\left(1 + 2\sqrt{\frac{t}{d}} + \frac{2t}{d} \right)\right\}\right) \leq 2e^{-t}.
\end{align}
By choosing $t = d$ and $\mathcal{E}_{\alpha} = e^{20(1-w)^2d\Delta^{2} / \sigma}$, we obtain by the union bound that Assumption~\ref{assump:One} holds with:
\begin{align*}
    \Prob(\max_{s \in \{1,\ldots,S\}}\chi^{2}(\mu_{s} \mid \mid \mu_{s-1}) \geq \mathcal{E}_{\alpha} - 1) \leq 2Se^{-d} = \delta_{\alpha}, 
\end{align*}
so that $\delta_{\alpha}$ appearing in Theorem~\ref{thm:Main Theorem Concentration} becomes negligible for large $d$. In this case, $\Delta = \mathcal{O}(1)$ and $\mathcal{E}_{\alpha} =\mathcal{O}(1)$ if $\sigma = \mathcal{O}(d)$. Roughly speaking, the second condition holds provided the variance of $\nu_{s}$ decays with $d$. 

\section{Chemical-Reactor Application}\label{section: Application}
Having established a theoretical basis for SMC in general cut-Bayesian posteriors and for the misspecified computer model setting with a Gaussian conditional posterior distribution, we proceed to demonstrate the SMC method and the linear tempering variant on a real scientific problem using the model of an ethylene-oxide production reactor that was previously introduced (i.e., Figure \ref{fig:reactor-parts}). We use this modular system in conjunction with experimental and simulated high-fidelity data in order to perform cut-Bayesian inference. Our main result is that both the SMC method of Theorem 1 and the linear tempering variant (Corollary 1) produce similar samples of calibration parameters to the ``gold standard" \citep{plummer2015cuts} direct-sampling method, while expending much less computational time.

The reactor model is a multi-physics system, coupling models from physics, chemistry, and engineering, as well as models of varying complexity. Production output is a steady numerical solution over a 1-dimensional spatial domain of single tube of a reactor filled with a fixed bed of catalyst pellets. The reactor concept is illustrated in Figure \ref{fig:reactor-op}.
\begin{figure}
    \centering
    \includegraphics[width=0.75\textwidth]{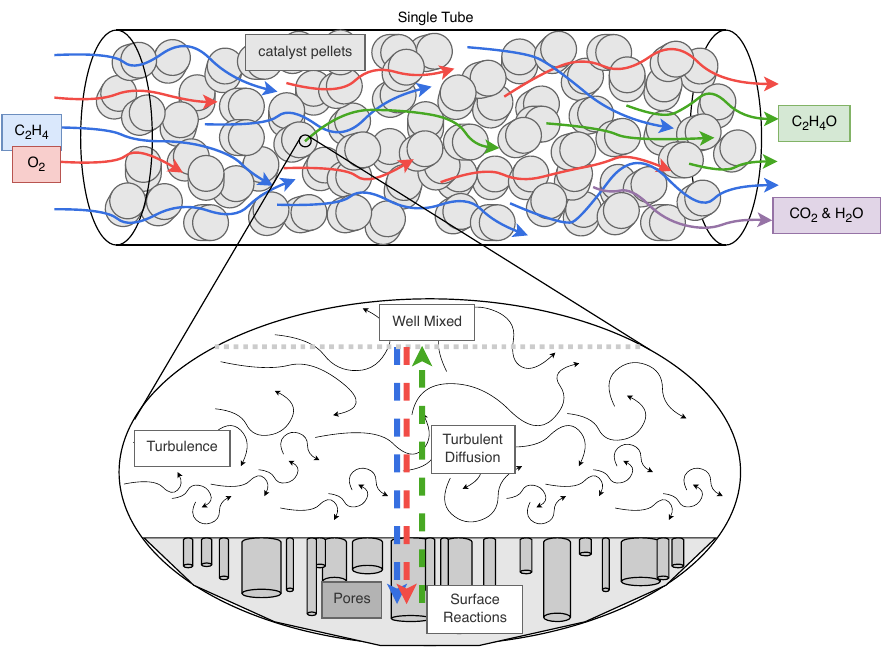}
    \caption{Representation of the catalytic reactor model for ethylene-oxide production. Input gas reactants are passed over a fixed catalytic bed, resulting in production of output products \& byproducts, varying in production efficiency. In addition to reaction parameters, the interaction parameters related to catalyst surface and flow turbulence are incorporated.}
    \label{fig:reactor-op}
\end{figure}
This model was chosen and constructed to be an extensible problem of a modular system demonstrating properties of complex system models, partitioning the integrated system into various submodels. The reactor is a common industrial process converting feedstocks ethylene and oxygen to ethylene oxide, in a reaction promoted by contact with a solid catalyst packing. The efficiency of the reactor is a function of system controls such as flow rates and temperature and of properties such as reaction rates and catalyst properties.

A brief description of the physical modeling and mathematical solution algorithm for these processes is provided.  First, at the largest scale, the complex packing of the pellets into the tube results in an uncertain void fraction, $\varepsilon$. The bulk of the gas flowing through these voids, at a given axial length down the tube, mixes quickly and is then considered to be well approximated by a single chemical composition -- this is referred to as the {\em plug-flow} approximation. This approach allows a coupled set of steady, one-dimensional differential equations (one equation for each compound) to model the change in chemical composition along the length of the reactor. The coupling among these equations is due to the chemical reaction rate converting the reactant compounds of ethylene and oxygen to either ethylene oxide or carbon dioxide \& water vapor. So next at the smallest scales, these reactions occur on the surface of the catalyst and are often limited by the available catalyst area. For this reason, highly porous catalyst pellets are preferred. The diagnostics for the properties of these pores such as the pore diameter, $D$, are trusted but are less than ideal \citep{Osterrieth2022BETreproducability}. The reaction-rate laws appear mathematically as algebraic expressions that depend on their own parameters, $C_{p1}$ and $K_{e1}$, and present the most probable point of model misspecification \citep{Klugherz1971KineticsOE, pu2019ethoxmechanisms}. This submodel for the reaction rate depends on the local chemical composition and is thus calculated separately at each discretized point along the length of the reactor. Finally at the intermediate scales, the amount of reactants at the surface is continually depleted by the reaction which also results in an increased concentration of products at the surface, which must diffuse to the bulk -- all resulting in a chemical composition local to the catalyst surface that may differ significantly from the composition in the neighboring bulk. Hence, the intermediate physical processes of pore reaction \& diffusion as well as turbulent diffusion in the millimeter-scale layer surrounding the pellet must both be modeled \citep{koning2002packedbedtransport}. The additional submodel for the turbulence around the pellet introduces two empirical parameters of $c$ and $n$. Unfortunately, the coupling of these submodels from the bulk to the surface requires an iterative multivariable algebraic solve at each discretized point along the reactor -- resulting in a significant increase of computational time.  The code for simulating this full model is implemented as a \texttt{Python} package that we utilized for demonstrating cut-Bayesian inference and is available at reference ~\cite{smith2023}. 

\subsection{Calibration and Cut Parameters}\label{subsection: params}
The parameters of the ethylene-oxide production reactor can be grouped by the various submodels included in Figure \ref{fig:reactor-parts}. For the purposes of this example, we first performed a Sobol sensitivity analysis \citep{sobol2001global} in order to select a set of parameters for which ethylene-oxide production is most sensitive. Based on this analysis and the availability of experimental data to calibrate submodels, we chose to focus on the turbulence, reaction rate, and catalyst submodels. The calibration parameters of interest are:
\begin{itemize}
\item $c$ and $n$: These parameters describe the relationship between Reynolds number of flow velocity and the turbulent Nusselt number, based on a linear model in a log-log scale.
\item $C_{p1}$ and $K_{e1}$: These are reaction rate parameters that determine the rate of ethylene-oxide production, using a rate law provided by \cite{Klugherz1971KineticsOE}. 
\end{itemize}
Further, we assume that the cut parameters come from the catalyst submodel, whose distributions are fixed based on domain expertise of catalyst properties:
\begin{itemize}
\item $\varepsilon$: Describes the gaps between pellets in the reactor. 
\item $D$: Diameter of the pores in the reactor. 
\end{itemize}
We designated these catalyst parameters as cut parameters because their uncertainties are well-characterized based on domain expertise. We also used nominal values based on domain expertise for any additional parameter values. Calibration and cut parameters are summarized in Table \ref{table:reactor_params}. For calibrated parameters, the distribution listed is the prior. For the cut parameters, the distribution is fixed and is not updated.

\begin{table}[h]
    \centering
    \begin{tabular}{|l|l|l|l|}
    \hline
    submodel & parameter & description & distribution \\
    \hline
        Turbulence  & $c$ & calibrated & U(.01,3)\\
                    & $n$ & calibrated & U(-.8,-.05) \\
        Reaction    & $C_{p1}$ & calibrated & U(0,.1) \\
                    & $K_{e1}$ & calibrated & U(0,.1) \\
        Catalyst    & $D$ & cut & U(.019, .021) \\
                    & $\varepsilon$ & cut & U(.6375, .8625) \\
    \hline
    \end{tabular}
    \caption{Parameters used for the reactor model demonstration. }
    \label{table:reactor_params}
\end{table}
We show illustrations of the turbulence submodel and reaction submodel data in Figure \ref{fig:exp_dat} below. 

\subsection{Model Specification}\label{subsection: model}
Data from actual experiments suggested by domain expertise were used for the turbulence and reaction submodels. These data were compiled from many classical experiments that are described and are available in the \texttt{Python} package from \cite{smith2023}. Examples of this data are illustrated in Figure \ref{fig:exp_dat}.

\begin{figure}
    \centering
    \includegraphics[width=0.8\textwidth]{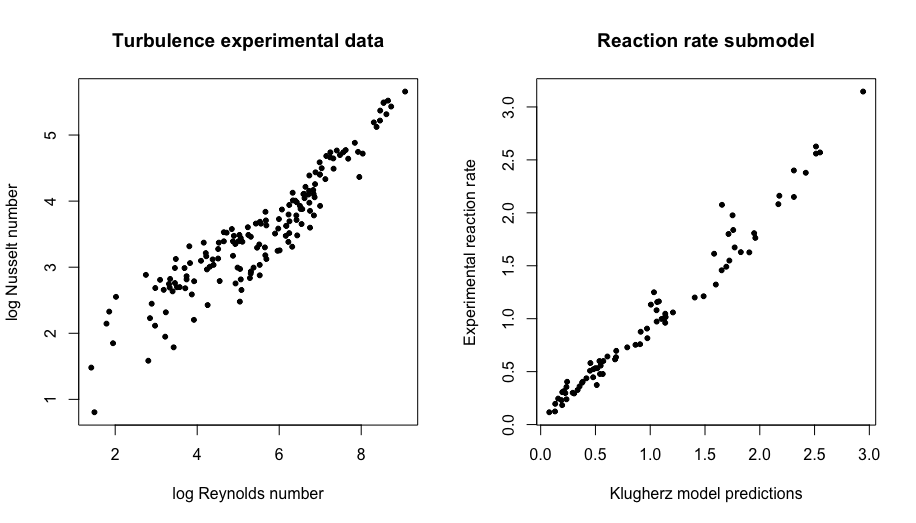}
    \caption{Left plot illustrates data from turbulence experiments, while right plot illustrates reaction rate data in comparison to the Klugherz reaction rate model evaluated at nominal parameter values from \citet{smith2023}. These data were used for the calibrations presented in Section \ref{subsection: results}. Log is the natural logarithm.}
    \label{fig:exp_dat}
\end{figure}

Let $x_1,..., x_M$ be the univariate turbulence data, $\mathbf{y}_1,..., \mathbf{y}_N$ be the bivariate experimental reaction rate data (one component for ethylene oxide and the second for carbon dioxide), and $z_1, ..., z_P$ be the univariate integrated ethylene-oxide production data generated with a high-fidelity model. Placed in a vector, we use the notation $\mathbf{x}$, $\mathbf{y}$, and $\mathbf{z}$. Further, $f_1(c,n, i)$ represents the simulator output for the turbulence submodel for the $i$th input setting, $\mathbf{f}_2(C_{p1}, K_{e1}, i)$ represents the simulator output for the reaction rate submodel at the $i$th input setting, and $f_3(c,n,C_{p1},K_{e1},\varepsilon, D, i)$ represents the simulator output for the integrated ethylene-oxide production model at the $i$th input setting. Tacitly, the $i$th input setting represents any inputs additional to the parameters that corresponded to the $i$th experimental observation; for instance, Reynolds number for the turbulence model, and temperatures and pressures for the rate kinetics model.

Conditional on the calibration and cut parameter values, we assume that each of the three data sets are independent and observed with independent, 0 mean normal errors with known variances. Estimates of the error variances for each datatype were determined with preliminary simple linear regression fits; we focus in this example solely on the inference of physical parameters. We denote $\sigma^2_1$ as the variance for the turbulence data, $\Sigma_2$ as the 2-by-2 covariance matrix for the reaction rate data, and $\sigma^2_3$ as the variance for the integrated ethylene-oxide production data.

Thus, we may write the likelihood, conditional on the calibration and cut parameters and simulators, as:
\begin{multline}
p(\mathbf{x}, \mathbf{y}, \mathbf{z} | c,n,K_{e1},C_{p1},\varepsilon,D) = \\
\prod_{i = 1}^M \mathcal{N}_{1}(x_i;f_1(c,n,i), \sigma^2_1) \times \prod_{i = 1}^N\mathcal{N}_{2}(\mathbf{y}_i;\mathbf{f}_2(C_{p1},K_{e1},i), \Sigma_2) \times \prod_{i = 1}^P\mathcal{N}_{1}(z_i;f_3(c,n,C_{p1},K_{e1},\varepsilon, D,i), \sigma^2_3).
\end{multline}
The uniform prior distributions for the calibration parameters $c, n, K_{e1}$, and $C_{p1}$  and the uniform cut distributions for $\varepsilon$ and $D$ are specified in Table \ref{table:reactor_params}.

\subsection{Results}\label{subsection: results}
We first implemented an SMC routine based off of Section \ref{subsection:cut-Bayes SMC} that uses independent draws of the cut parameter. We used 25 particles with 5 slice sampling \citep{neal_slice} mutate steps within the SMC algorithm. We grouped 10 independently drawn cut parameters in a single batch and ran 8 independent batches, and pooled all of the resampled and mutated particles from the 8 distinct batches together. To compare the resultant cut-Bayesian posterior samples of $c, n, C_{p1}$, and $K_{e1}$, we generated ``gold standard" samples with the direct-sampling method discussed in Section \ref{section: cut-Bayes}. Specifically, conditional on each draw of cut parameters we use the slice-within-Gibbs \citep{neal_slice} algorithm to draw from the conditional posterior of calibration parameters, running for a total of 1000 iterations; 1000 was chosen based on examining traceplots of preliminary runs as well as calculating the Gelman-Rubin convergence diagnostic \citep{gelman1992inference} using the implementation from the \texttt{coda} package \citep{plummer2006coda} for said runs. Pooling all of the calibration samples generated from distinct cut parameter draws results in samples from the cut posterior distribution \citep{plummer2015cuts}. In Figure \ref{fig:SMC_compare} we see a comparison of cut posterior densities for calibration parameters generated with the SMC method and the direct-sampling method (i.e., using multiple MCMC chains). We see that the densities generally match up, suggesting the SMC method is producing samples of the cut posterior. Appendix \ref{app: pairs-plot} illustrates that pairs plots of samples from the SMC and direct methods show agreement; moreover, the bivariate distributions of pairs indicate curvature that is not characteristic of a bivariate normal distribution. 

In an attempt to reduce computational runtime, we additionally implemented the linear tempering variant of Section \ref{subsection: temperedSMC} with 10 particles and 1 additional point on the line connecting successive cut parameter draws, with 5 independently drawn cut parameters per set. We also reduced the number of mutate steps by 1. We also ran a comparison with the travelling salesman permutation technique introduced in Section \ref{subsection: permutedSMC}, with 10 particles, 10 independently drawn cut parameters, and 4 mutation steps per job, using the \texttt{TSP} R package \citep{TSPref} to permute the ordering of the cut parameters in each job as described in Section \ref{subsection: permutedSMC}. The results presented are using Euclidean distance between cut parameters, though a version scaled so that the cut parameters are uniform between 0 and 1 was also compared. Very minor differences were observed in the resultant samples. The comparison of resultant samples to the the gold standard is given in the density plots of Figures \ref{fig:temper_compare} and \ref{fig:permute_compare}, also showing general agreement, and as is next discussed, indicating computational improvement in runtime. 

\begin{figure}
    \centering
    \includegraphics[width=0.8\textwidth]{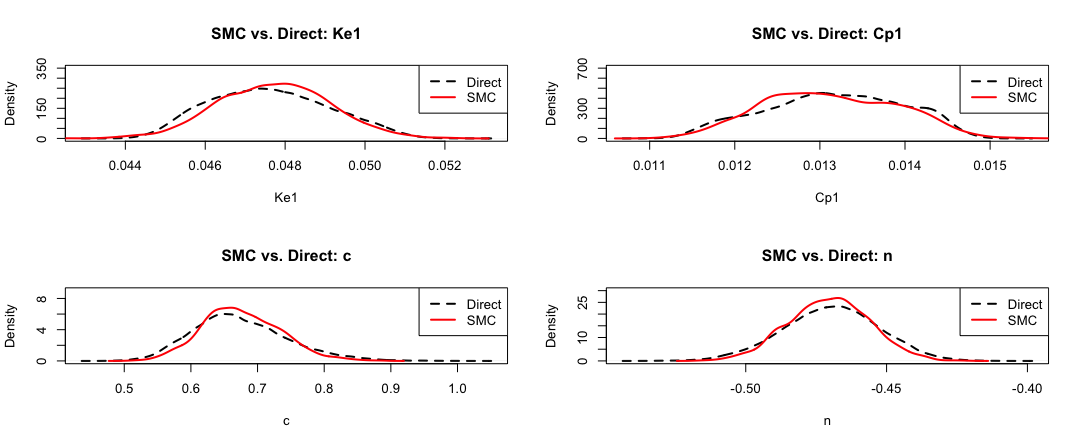}
    \caption{Comparison of SMC cut posterior samples (red density) of calibration parameters to direct-sampling (black density); the densities generally appear to be consistent. Additionally, Appendix \ref{app: pairs-plot} contains a comparison of pairs plots, which show alignment.}
    \label{fig:SMC_compare}
\end{figure}

\begin{figure}
    \centering
    \includegraphics[width=0.8\textwidth]{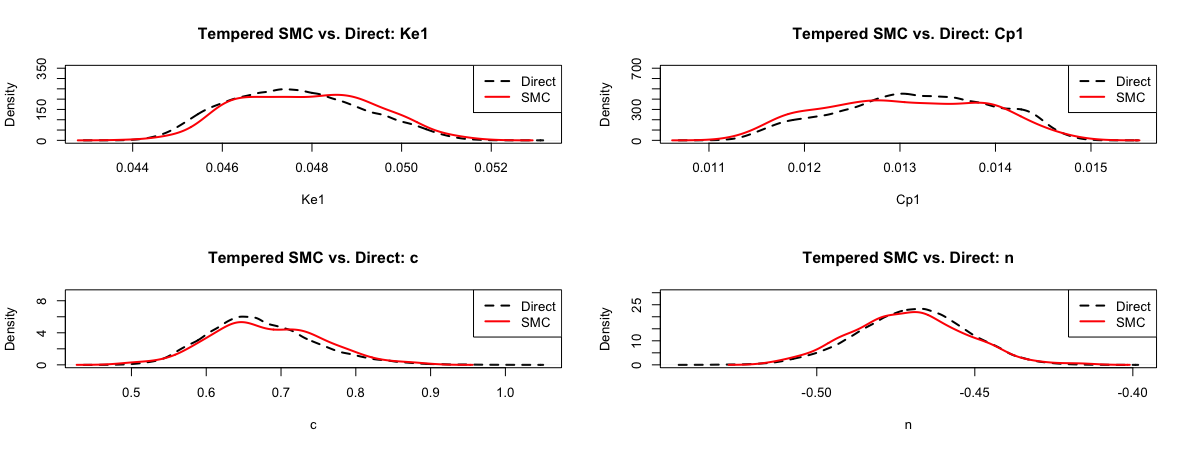}
    \caption{Comparison of tempered SMC cut posterior samples (red density) of calibration parameters to direct-sampling (black density). 10 particles are used for this run in comparison to the 25 used in Figure \ref{fig:SMC_compare}, but with $P =1$, as in Corollary 1. Samples appear to be generally consistent with direct sampling approach but with substantial computational reduction.}
    \label{fig:temper_compare}
\end{figure}

\begin{figure}
    \centering
    \includegraphics[width=0.8\textwidth]{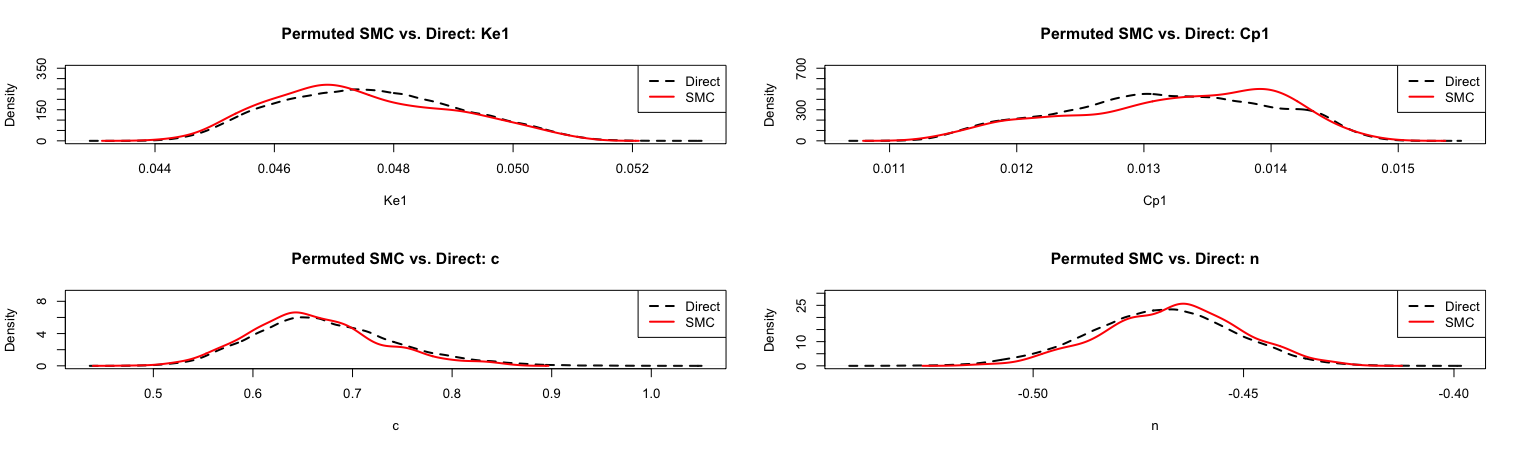}
    \caption{Comparison of permuted SMC cut posterior samples (red density) of calibration parameters to direct-sampling (black density). 10 particles are used for this run just as the tempered case, and 10 total cut parameters as in the previous two cases. Quality of SMC samples appears similar to the tempered case.}
    \label{fig:permute_compare}
\end{figure}

Both SMC procedures and the direct-sampling procedure using multiple Markov chains were run on the Los Alamos National Laboratory (LANL) Darwin computing testbed, which is funded by the Computational Systems and Software Environments subprogram of LANL's Advanced Simulation and Computing program (NNSA/DOE). All methods were batched into 8 jobs where each job was run on its own node on the general partition, but no parallelization was used beyond batching independent runs. Additionally, the BASS R package \citep{JSSv094i08} was used to construct a computationally efficient emulator for the integrated ethylene production model (referred to as $f_3$ previously) and the identical emulator was used in the SMC and the direct-sampling procedures. Run time results comparing the SMC method, the linear tempered SMC method, and the direct sampling MCMC method are shown in Table \ref{table:results}, indicating substantial improvement using the SMC methods. The overhead time needed to generate the initial particles required for the SMC methods was 1 hour and 19 minutes. 

\begin{table}[h]
    \centering
    \begin{tabular}{|l|l|l|l|}
    \hline
    Method & Min job runtime (minutes) & Max job runtime (minutes) \\
    \hline
        Direct sampling  & 645 & 765\\
        SMC    & 79 & 92 \\
        Tempered SMC   & 29  & 39 \\
        Permuted SMC & 29 & 46 \\
    \hline
    \end{tabular}
    \caption{Comparison of run times in minutes for direct sampling, SMC (Section \ref{subsection:cut-Bayes SMC}), the linear tempered SMC variant (Section \ref{subsection: temperedSMC}), and the permuted SMC variant (Section \ref{subsection: permutedSMC}) over 8 jobs run in parallel. Densities are formed from samples pooled together from 8 independently run jobs.}
     \label{table:results}
\end{table}
Some considerations should be underscored in interpreting these comparisons. First, the direct-sampling method's run time can be reduced by running the MCMC chains for less than 1000 iterations; 1000 was chosen based on a convergence diagnostic. It is nonetheless plausible that the SMC method saves time because each set of Markov transitions should ideally start from a conditional posterior distribution that is closer to the next distribution sampled from instead of starting from scratch as in the MCMC case.

\section{Conclusion and Discussion}\label{section: conclusion}
Our central contribution is the introduction of SMC methods for computing with cut-Bayesian posteriors. The methods are supported by theoretical finite-sample concentration results of SMC estimators in general settings as well as finite sample complexity bounds when the conditional posterior is normally distributed, as was motivated by a computer model misspecification problem. Additionally, we have demonstrated the practical utility and accuracy of the methods for cut-Bayes inference with a coupled modular chemical reactor system, and have shown order-of-magnitude computational efficiency gains. To our knowledge, we have presented the first provably correct SMC-based computational method for cut-Bayesian posterior inference. The defining characteristics of our method are that it results in convergent estimators and does not require approximations to conditional distributions. When such approximations are appropriate, alternate methods may be preferred for computational efficiency; however, our SMC method can be used when such approximations cannot be made (see Appendix \ref{app: non-Gauss} for an example in which the conditional posteriors are non-Gaussian). While the unbiased MCMC method of \cite{jacob2020} could be a viable alternative that does not assume approximations to conditional distributions, care must be taken to ensure that meeting times of coupled Markov chains are not prohibitively expensive, such as when the dimension of the problem increases. The results of Section \ref{section: cutBayesex} provide some theoretical guidance for how the complexity of our SMC method scales with dimension for a specific case. 

There are a number of future research directions that can build off of the work herein. The problems encountered that have motivated this work are where a full cut, instead of a partial cut, is most sensible. In other words, some experiments and models are trusted to inform the posterior of key physical parameters, and experiments from poorly specified models should not be allowed to contaminate said inferences at all. However, in some regimes, it is possible that a model is not so poorly specified that some feedback is warranted; this situation is known as semi-modular inference \citep{carmona2020semi}, and it is possible that the SMC methods developed in this paper can extend to such a setting. Future work can build off of the contributions here for inference in the semi-modular setting. Additionally, an important aspect of our theoretical results is that independence is assumed for the cut parameter draws for the application of Hoeffding's inequality; while concentration results are less prevalent in the case of dependence, future work could explore relaxing the independence assumption in order to derive tighter concentration bounds.

\section{Acknowledgments}
The authors would like to thank an Associate Editor and two Reviewers for their comments that have helped to improve the manuscript.

This work was supported by the U.S. Department of Energy through the Los Alamos National Laboratory. Los Alamos National Laboratory is operated by Triad National Security, LLC, for the National Nuclear Security Administration of U.S. Department of Energy (Contract No. 89233218CNA000001). The authors are grateful for support from the Advanced Simulation and Computing Program’s Verification and Validation subprogram.

\appendix
\section{Proof of Theorem~\ref{thm:Main Theorem Concentration}}\label{app:A}

\subsection{General Results for SMC Algorithms}\label{app:General Results}

We first prove a straightforward extension of the finite sample bounds for SMC algorithms given in \cite{marion}. Let $\mu_{0},\ldots,\mu_{S} = \pi$ denote a sequence of distributions defined on a space $\mathcal{X}$ (e.g. $\mathcal{X} = \mathbb{R}^{d}$) with common dominating measure $\rho$. The sequential Monte Carlo algorithm produces a set of particles $\Theta = (\h^{1:N}_{0},\tilde{\h}^{1:N}_{1},\h^{1:N}_{1},\ldots,\tilde{\h}^{1:N}_{S},\h_{S}^{1:N})$ jointly defined on $\mathcal{X}^{N(2S+1)}$. The marginal laws of these random variables are implicitly defined by the algorithm given in Section~\ref{subsection:cut-Bayes SMC}; we do not state these marginal laws here for brevity but refer the reader to \cite{marion} and \cite{mathews} for details. In addition, the proof technique of \cite{marion}, and the proof technique used here to extend their results, involves constructing a set of random variables denoted $\bar{\h}^{1:N}_{s}$ via a ``maximal" coupling construction at the $s$th step of the algorithm so that $\bar{\h}^{i}_{s} \sim \mu_{s}$ and $\bar{\h}^{i}_{s}$ and $\h^{i}_{s}$ are equal with high probability. The state space we consider throughout is $\mathcal{X}^{N(3S+1)}$ with corresponding product $\sigma$-field $\mathcal{B}^{N(3S+1)}$ corresponding to the particle system $\Theta$ along with the constructed random variables $\bar{\Theta} = (\bar{\h}^{1:N}_{1},\ldots,\bar{\h}^{1:N}_{S})$. We let $\Prob$ and $\E$ generically denote the joint probability measure and expectation, respectively, of  $\Theta$ and $\bar{\Theta}$ (i.e. the probability space is $(\mathcal{X}^{N(3S+1)},\mathcal{B}^{N(3S+1)},\Prob )$). 

We note that in practice $(\bar{\h}^{1:N}_{1},\ldots,\bar{\h}^{1:N}_{S})$ are not constructed during the SMC algorithm; here they are only constructed in theory for the proof technique. We refer the reader to the Appendix of \cite{marion} for details on how the construction of $\bar{\Theta}$ is performed along with some important properties of these random variables. In particular, the authors show that $\bar{\h}^{i}_{s} \indep \tilde{\h}^{i}_{s}$ and, consequently, $\bar{\h}^{1:N}_{s} \overset{iid}{\sim} \mu_{s}$ and 
\begin{align}\label{eqn: bar cond independent}
    \bar{\theta}^{1:N}_{s} \indep (\h^{1:N}_{0},\theta^{1:N}_{1},\theta^{1:N}_{2},\tilde{\theta}^{1:N}_{2}\bar{\theta}^{1:N}_{2},\ldots,\theta^{1:N}_{s-1},\tilde{\theta}^{1:N}_{s-1}\bar{\theta}^{1:N}_{s-1}) 
\end{align}
for $s=1,\ldots,S$. We use both of these properties below. As in \cite{marion}, let
\begin{align*}
    A_{s} &= \{\theta^{1:N}_{s} = \bar{\theta}^{1:N}_{s} \} \\
    B_{s} &= \{\bar{w}_{s+1}   \geq 2\E_{\mu_{s}}[w_{s+1}] / 3 \} \\
    C_{s} &= A_{s} \cap B_{s}, 
\end{align*}
where $\bar{\theta}^{1:N}_{0} := \theta^{1:N}_{0} \overset{iid}{\sim} \mu_{0}$. As in Section \ref{subsection:cut-Bayes SMC}, $\bar{w}_{s+1} := \frac{1}{N}\sum^{N}_{i=1} w_{s+1}(\theta^{i}_{s})$, where $w_{s+1}$ are importance weights. \cite{marion} show inductively that $\Prob(C_{s})$ can be made large for $s = 0,\ldots,S$ by choosing $N$ and $t$ sufficiently large. However, here there is no ``target" measure and we wish to obtain samples from all distributions $\mu_{0},\ldots,\mu_{S}$. Therefore, we extend this result to show that the \textit{joint} event $\Prob(C_{0} \cap \ldots \cap C_{s})$ holds with high probability. Clearly a bound on $\Prob(C_{s})$ implies a bound on this joint probability via the union bound. However, doing so results in a looser bound than what is necessary. Throughout we make the following assumption:
\begin{assump}\label{assump: L2}
There exists $\mathcal{E} > 0$ such that
\begin{align*}
    \max_{s \in \{1,\ldots,S\}} \chi^{2}(\mu_{s} \mid \mid \mu_{s-1}) \leq \mathcal{E} \text{ for } s = 1,\ldots,S.
\end{align*}
\end{assump}
\noindent
We slightly modify the proof given in \cite{marion} to show the following result.
\begin{lemma}\label{lemma:Finite Sample Lemma}
Let $\delta \in (0,\frac{1}{4})$ and choose
\begin{enumerate}
    \item $N \geq 18 \log\left(\frac{2(S+1)}{\delta} \right)  (\mathcal{E} + 1)$
    \item $t \geq \max_{s=1,\ldots,S} \tau_{s}(\frac{\delta}{2N(S+1)}, 2)$
\end{enumerate}
Then for $s \in \{0,\ldots,S\}$:
\begin{align*}
    \Prob(C_{0},\ldots,C_{s}) \geq \left(1-\frac{\delta}{S+1}\right)^{s+1}.
\end{align*}
\end{lemma}
The proof of Lemma~\ref{lemma:Finite Sample Lemma} proceeds by induction and the arguments are near identical to those given in \cite{marion}. Consequently, we will often point to results given there. The following lemma establishes the base case. 
\begin{lemma}\label{lemma:Base Case}
Let $N \geq 18 \log\left(\frac{S+1}{\delta} \right)  (\mathcal{E} + 1)$. Then $\Prob(C_{0}) \geq 1-\frac{\delta}{S+1}$.
\end{lemma}
\begin{proof}
Note $\Prob(C_{0}) = \Prob(B_{0})$ since $\bar{\theta}^{1:N}_{0} := \theta^{1:N}_{0}$. By assumption, $\theta^{1:N}_{0} \overset{iid}{\sim} \mu_{0}$. Hence, by Bernstein's inequality
\begin{align*}
    \Prob(B^c_{0}) &= \Prob\left(\frac{1}{N}\sum^{N}_{i=1} w_{1}(\theta^{i}_{0})   \leq  \frac{2\E_{\mu_{0}}[w_{1}]}{3}\right) \\
    &\leq \exp\left\{ - \frac{N}{18(1 + \chi^{2}(\mu_{1} \mid \mid \mu_{0}))} \right\} \leq \frac{\delta}{S+1},
\end{align*}
where the final inequality follows by our choice of $N$ and Assumption~\ref{assump: L2}.
\end{proof}

\begin{proof}(Lemma~\ref{lemma:Finite Sample Lemma}) 
The proof proceeds by induction. We use the shorthand notation $\Prob(C_{0:s}) := \Prob(C_{0},\ldots,C_{s})$ throughout. Lemma~\ref{lemma:Base Case} establishes the base case by our choice of $N$. Now suppose the claim holds at step $s-1$:
\begin{align*}
    \Prob(C_{0:s-1}) \geq \left(1-\frac{\delta}{S+1}\right)^{s}.
\end{align*}
Since $\delta \in (0,\frac{1}{4})$, it follows that $\Prob(C_{0:s-1}) \geq \frac{3}{4}$ by the induction hypothesis. This implies (see Lemma 4.2 of \cite{marion}) 
\begin{align*}
   \Prob(\tilde{\h}^{i}_{s} \in B \mid C_{0:s-1}) \leq 2 \mu_{s}(B).
\end{align*}
Consequently, $\mathcal{L}(\tilde{\h}^{i}_{s} \mid C_{0:s-1})$ is $2$-warm start with respect to $\mu_{s}$, where $\mathcal{L}(\tilde{\h}^{i}_{s} \mid C_{0:s-1})$ denotes the law of $\tilde{\h}^{i}_{s}$ conditional on the event $C_{0:s-1}$ occurring. Hence, by a coupling argument (see Lemma 4.1 and the Appendix of \cite{marion}), we can guarantee by our choice of $t$
\begin{align*}
    \Prob(A_{s} \mid C_{0:s-1}) \geq 1-\frac{\delta}{2(S+1)}.
\end{align*}
Let
\begin{align*}
    \bar{B}_{s} = \left\{\frac{1}{N}\sum^{N}_{i=1} w_{s+1}(\bar{\theta}^{i}_{s})   \geq 2\E_{\mu_{s}}[w_{s+1}]/3  \right\}.
\end{align*}
That is, $\bar{B}_{s}$ is the same as $B_{s}$, except we replace the $\theta^{1:N}_{s}$ particles with the constructed ``target" random variables $\bar{\theta}^{1:N}_{s} \overset{iid}{\sim} \mu_{s}$. As in the proof of Lemma~\ref{lemma:Base Case}, we have by Bernstein's inequality and our choice of $N$
\begin{align*}
    \Prob(\bar{B}^{c}_{s}) \leq \frac{\delta}{2 (S+1)}.
\end{align*}
Putting it all together,
\begin{align*}
   \Prob(C^{c}_{s} \mid C_{0:s-1}) &= \Prob(B^{c}_{s} \cap A_{s} \mid C_{0:s-1}) + \Prob(A^{c}_{s} \mid C_{0:s-1}) \\
   &\leq \Prob(\bar{B}^{c}_{s} \mid C_{0:s-1}) + \Prob(A^{c}_{s} \mid C_{0:s-1}) \\
   &=  \Prob(\bar{B}^{c}_{s}) + \Prob(A^{c}_{s} \mid C_{0:s-1})  \\
   &\leq \frac{\delta}{S+1}
\end{align*}
The equality $\Prob(\bar{B}^{c}_{s} \mid C_{0:s-1}) = \Prob(\bar{B}^{c}_{s})$ follows using \eqref{eqn: bar cond independent}
(see the Appendix of \cite{marion}). The stated bound follows:
\begin{align*}
    \Prob(C_{0},\ldots,C_{s}) = \Prob(C_{s} \mid C_{0},\ldots,C_{s-1} )\Prob(C_{0},\ldots,C_{s-1}) \geq \left(1 - \frac{\delta}{S+1} \right)^{s+1}.
\end{align*}
\end{proof}

We now use Lemma~\ref{lemma:Finite Sample Lemma} to prove Theorem~\ref{thm:Main Theorem Concentration}. We consider the same probability space as before, except expanded so that $\nu_{0:S} \overset{iid}{\sim} p_{\nu}$ (assumed to also be defined on $\mathcal{X}$) are jointly defined with $\Theta$ and $\bar{\Theta}$; we again let $\Prob$ and $\E$ denote the joint probability measure and expectation, respectively. We appeal to the fact that conditional on $\nu_{0:S}$, the sequence $\mu_{0},\ldots,\mu_{S}$ is fixed and the results above apply.

\begin{proof}(Theorem~\ref{thm:Main Theorem Concentration})
We have by our choice of $N$ and $t$ along with Lemma~\ref{lemma:Finite Sample Lemma} that 
\begin{align*}
    &\Prob\left(\theta^{1:N}_{0} = \bar{\theta}^{1:N}_{0}, \ldots, \theta^{1:N}_{S} = \bar{\theta}^{1:N}_{S} \mid \max_{s \in \{1,\ldots,S\}}\chi^{2}(\mu_{s} \mid \mid \mu_{s-1}) \leq \mathcal{E}_{\alpha} - 1 \right)  \geq 1-\frac{\delta}{3}.
\end{align*}
By the law of total probability and Assumption~\ref{assump:One}
\begin{align}\label{eqn:Joint Coupling Unordered}
    \Prob(\theta^{1:N}_{0} = \bar{\theta}^{1:N}_{0}, \ldots, \theta^{1:N}_{S} = \bar{\theta}^{1:N}_{S}) \geq 1-\frac{\delta}{3} - \delta_{\alpha}.
\end{align}
Going forward we work with the $(\bar{\h}^{1:N}_{0},\ldots,\bar{\h}^{1:N}_{S})$ random variables instead of $(\h^{1:N}_{0},\ldots,\h^{1:N}_{S})$ and then appeal to \eqref{eqn:Joint Coupling Unordered} to obtain the stated claim. Let $\mu_{g}(\nu) := \E_{\pi(\theta \mid \y, \nu)}[g(\nu,\h)]$ denote the mean of $g(\nu,\h)$ with respect to the distribution $\pi(\h \mid \y, \nu)$ (note $\mu_{g}(\nu)$ is a random variable). Recall that conditional on $\nu_{0},\ldots,\nu_{S}$, the sequence $\mu_{0},\ldots,\mu_{S}$ is fixed and so $\bar{\theta}^{1:N}_{s} \mid \nu_{0:S} \overset{iid}{\sim} \pi(\cdot \mid \y, \nu_{s})$ and
\begin{align*}
    \E[g(\nu_{s},\bar{\h}^{i}_{s}) \mid \nu_{0:S} ] = \mu_{g}(\nu_{s}), \text{ for }i=1,\ldots,N.
\end{align*}
By the conditional form of Hoeffding's inequality (conditional on $\nu_{0:S}$), we have by our choice of $N$ that with probability $1 - \frac{\delta}{3}$
\begin{align}\label{eqn:first concentration}
   \left| \frac{1}{N}\sum^{N}_{i=1} g(\nu_{s},\bar{\theta}^{i}_{s}) - \mu_{g}(\nu_{s})\right| \leq \frac{\epsilon}{2}, \text{ for }s=0,\ldots,S.
\end{align}
Since \eqref{eqn:first concentration} holds jointly for $s=0,\ldots,S$, this implies that with probability $1 - \frac{\delta}{3}$ by the triangle inequality that conditional on $\nu_{0:S},$
\begin{align}\label{eqn:Theta Concentration Event}
    \left|\frac{1}{S+1}\sum^{S}_{s=0} \left(\frac{1}{N}\sum^{N}_{i=1} g(\nu_{s},\bar{\theta}^{i}_{s})\right) - \frac{1}{S+1}\sum^{S}_{s=0} \mu_{g}(\nu_{s})\right| \leq \frac{\epsilon}{2}.
\end{align}
Recall that $\nu_{0:S} \overset{iid}{\sim} p_{\nu}$ and note that the random variable $|\mu_{g}(\nu_{s})| \leq 1$ since $|g| \leq 1$. Hence, we again have by our choice of $S$ that with probability $1-\frac{\delta}{3}$ (by Hoeffding's inequality),
\begin{align}\label{eqn:Nu Concentration Event}
   \left| \frac{1}{S+1}\sum^{S}_{s=0} \mu_{g}(\nu_{s}) - \E_{p_{\nu}}[\mu_{g}(\nu)] \right| \leq \frac{\epsilon}{2}.
\end{align}
By definition, $\E_{p_{\nu}}[\mu_{g}(\nu)] = \E_{\pi^{\text{cut}}}[g]$, the expectation under the cut-Bayes posterior. Consequently, \eqref{eqn:Theta Concentration Event} and \eqref{eqn:Nu Concentration Event} along with the triangle inequality imply that with probability $1 - \frac{2\delta}{3}$
\begin{align}\label{eqn:Double Concentration}
    \left|\frac{1}{S+1}\sum^{S}_{s=0} \left(\frac{1}{N}\sum^{N}_{i=1} g(\nu_{s},\bar{\theta}^{i}_{s})\right) - \E_{\pi^{\text{cut}}}[g]\right| \leq \epsilon.
\end{align}
Now, \eqref{eqn:Double Concentration} holds only for the $(\bar{\h}^{1:N}_{0},\ldots,\bar{\h}^{1:N}_{S})$ random variables. However, by appealing to \eqref{eqn:Joint Coupling Unordered} it follows that with probability $1-\frac{\delta}{3} - \delta_{\alpha} $ the same concentration inequalities hold for the $(\h^{1:N}_{0},\ldots,\h^{1:N}_{S})$ too. Hence, by \eqref{eqn:Joint Coupling Unordered} and \eqref{eqn:Double Concentration} it follows that with probability $1 - \delta - \delta_{\alpha}$,
\begin{align*}
    \left| \hat{g} - \E_{\pi^{\text{cut}}}[g] \right| \leq \epsilon.
\end{align*}
\end{proof}

\newpage
\section{Pairs Plots for SMC and Direct Sampling in Reactor Example}\label{app: pairs-plot}
\begin{figure}[ht]
    \centering
    \includegraphics[width=0.35\textwidth]{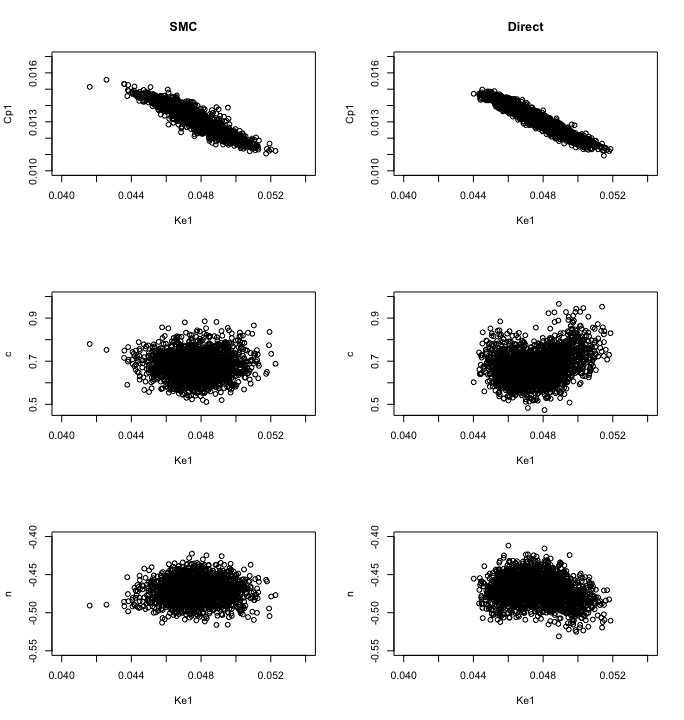}
    \caption{SMC cut posterior samples pairs plot versus direct sampling pairs plot for Ke1 pairs.}
    \label{fig: pairs1}
\end{figure}

\begin{figure}[ht]
    \centering
    \includegraphics[width=0.35\textwidth]{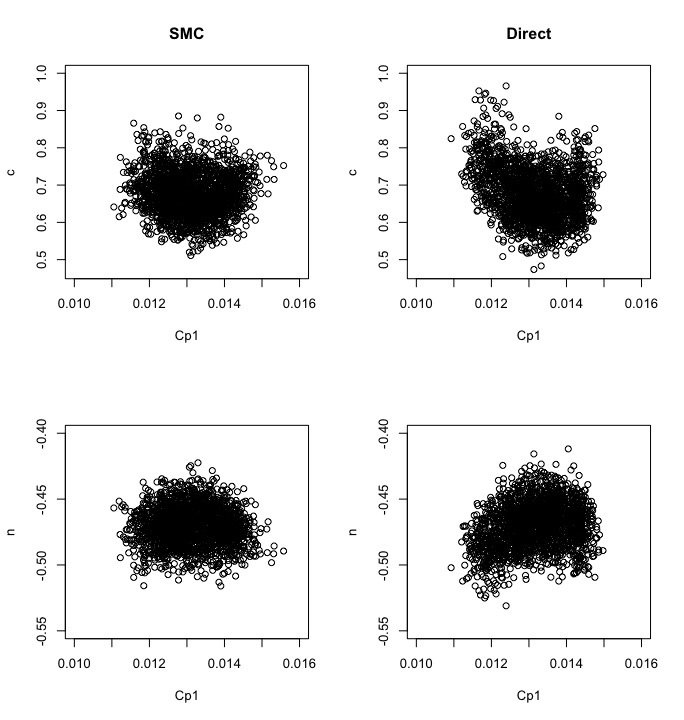}
    \caption{SMC cut posterior samples pairs plot versus direct sampling pairs plot for Cp1 pairs.}
    \label{fig: pairs2}
\end{figure}

\begin{figure}[ht]
    \centering
    \includegraphics[width=0.35\textwidth]{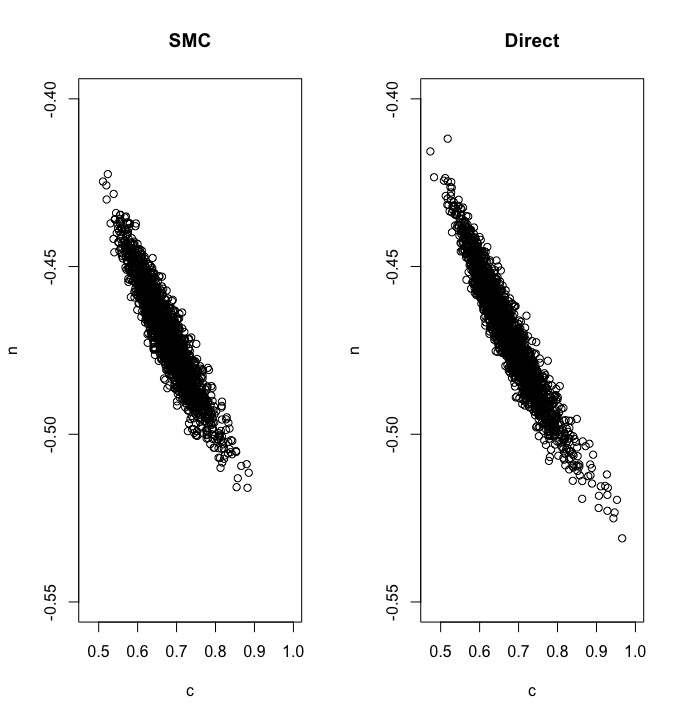}
    \caption{SMC cut posterior samples pairs plot versus direct sampling pairs plot for c,n pair.}
    \label{fig: pairs3}
\end{figure}

\newpage
\section{Non-Gaussian Conditional Posterior} \label{app: non-Gauss}
Here we provide a system which illustrates that conditional posteriors in cut Bayesian modeling can be non Gaussian, particularly with nonlinear functions. We consider the calibration parameters $\theta := (\theta_1, \theta_2) \in \mathbb{R}^2$ and cut parameter $\nu \in \mathbb{R}^1$. The data model $\y|\theta, \nu$ is given by $\mathcal{N}_{2}(f(\theta, \nu), \textrm{diag}(.1,1))$, and uniform, independent priors between $-30$ and $30$ are assumed for the components of $\theta$. Further $f: \mathbb{R}^3 \rightarrow \mathbb{R}^2$ is defined to be:

\begin{align*}
    f(\theta_1, \theta_2, \nu) := (\sin(\theta_1)\cos(\theta_2)\tan(\nu), \theta_1^2 + \theta_2^2 + \nu^2).
\end{align*}

We simulate a single realization of $\y$ according to a true value of $\theta$ equal to (1,2) and $\nu$ equal to 1. The conditional posterior, sampled via slice-within-Gibbs sampling, appears in Figures \ref{fig: nonGauss1} and \ref{fig: nonGauss2} below for values of $\nu = 1$ and $\nu = .3$ respectively, clearly exhibiting non-Gaussian behavior. This example is meant to illustrate that a normal approximation for the conditional posterior distribution is not always tenable. 

\newpage

\begin{figure}[h]
    \centering
    \includegraphics[width=0.45\textwidth]{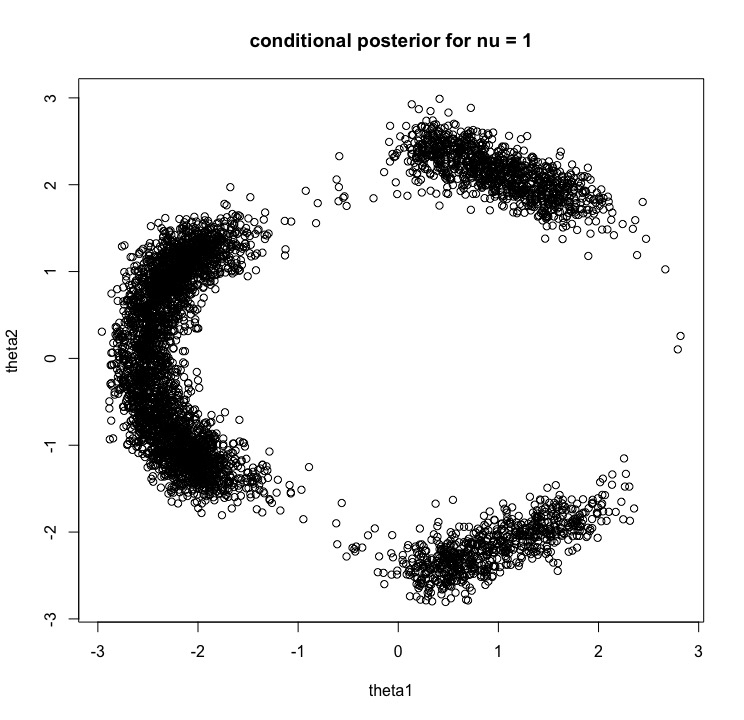}
    \caption{Non-Gaussian conditional posterior for $\nu = 1$.}
    \label{fig: nonGauss1}
\end{figure}

\begin{figure}[h]
    \centering
    \includegraphics[width=0.45\textwidth]{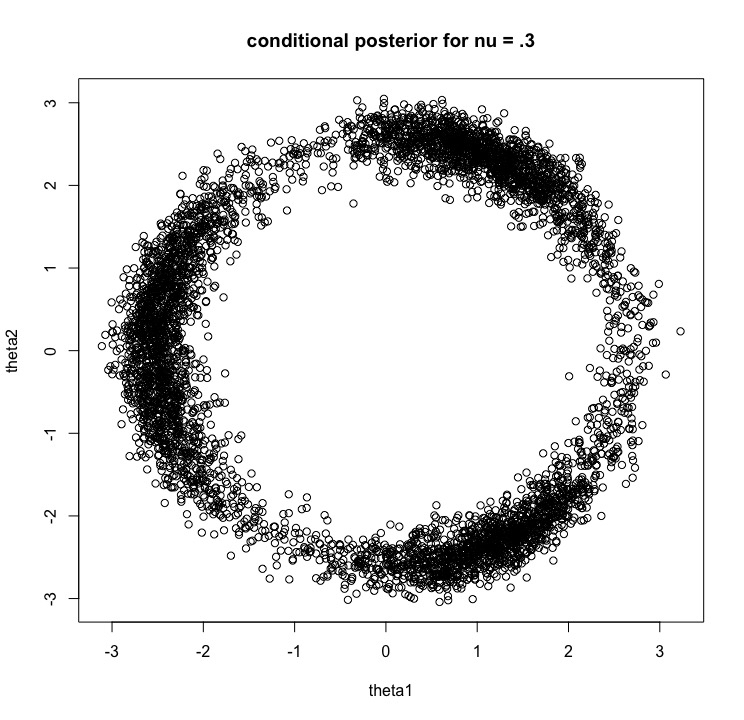}
    \caption{Non-Gaussian conditional posterior for $\nu = .3$.}
    \label{fig: nonGauss2}
\end{figure}

\newpage
\bibliographystyle{apalike}
\bibliography{references}

\end{document}